\documentclass[journal]{IEEEtran}
\ifCLASSINFOpdf
\else
\fi
\pdfoutput=1

\usepackage[T1]{fontenc}
\usepackage{graphicx}
\usepackage{caption}
\usepackage{subcaption}
\usepackage{amssymb}
\usepackage{amsmath}
\usepackage{amsthm}
\usepackage[noend]{algorithm}
\usepackage[noend]{algorithmic}
\usepackage{multirow}
\usepackage[percent]{overpic}

\newcounter{AlgorithmCounter}
\addtocounter{AlgorithmCounter}{1}

\newcommand{\FUNCTION}{\item[\textbf{Algorithm
 \arabic{AlgorithmCounter}\addtocounter{AlgorithmCounter}{1}:}]}
 
\graphicspath{img/}

\newcommand\InPUT{\hspace{1pt}I\kern-.05em
  {\kern-.05em%
    \lower-.3ex\hbox{n}}%
  \kern-.1em
  P\kern-.17em\hbox{U}\kern-.085em \lower-.009ex\hbox{T\hspace{2pt}}}

\newtheorem{prop}{Proposition}

\begin{document}
\title{End-to-End Reliability-Aware Scheduling for Wireless Sensor Networks}
%
\author{
Felix~Dobslaw,~
        Tingting~Zhang,
        and~Mikael~Gidlund

\thanks{F. Dobslaw, Tingting Zhang and Mikael Gidlund are with the Computer
Science Department, Mid-Sweden University, \"Ostersund and Sundsvall, Sweden
(e-mail: firstname.lastname@miun.se).}
\thanks{Manuscript received August 31, 2014; revised November 7, 2014;
accepted December  5, 2014.} }
%
%

\markboth{IEEE Transactions on Industrial Informatics,~Vol.~X, No.~Y,
MONTH~YEAR}%
{Dobslaw \MakeLowercase{\textit{et al.}}: Crosslayer Scheduling}
%



\maketitle

\begin{abstract}
Wireless Sensor Networks (WSN) are gaining popularity as a flexible and
economical alternative to field-bus installations for monitoring and control
applications.
For mission-critical applications, communication networks must provide
end-to-end reliability guarantees, posing substantial challenges for WSN.
Reliability can be improved by redundancy, and is often addressed on the MAC layer by re-submission of lost packets, usually applying slotted scheduling.
Recently, researchers have proposed a strategy
to optimally improve the reliability of a given schedule by repeating the most
rewarding slots in a schedule incrementally until a deadline. This \textit{Incrementer}
can be used with most scheduling algorithms but has scalability issues which
narrows its usability to offline calculations of schedules, for networks that are rather static. In
this paper, we introduce SchedEx, a generic heuristic scheduling
algorithm extension which guarantees a user-defined end-to-end reliability.
SchedEx produces competitive schedules to the existing
approach, and it does that consistently more than an \emph{order of magnitude}
faster. The harsher the end-to-end reliability demand of the network, the 
better SchedEx performs compared to the Incrementer.
We further show that SchedEx has a more
evenly distributed improvement impact on the scheduling algorithms, whereas the
Incrementer favors schedules created by certain scheduling algorithms.
\end{abstract}

%
\IEEEpeerreviewmaketitle
 
\section{Introduction}
\label{intro}
Wireless has been promoted as a flexible alternative to wired
field-bus, and is therefore pushing forward into the industrial market \cite{hancke2013special}, thereby widening the scope from more basic information
harvesting (sensing) applications for the offline analysis of data, to receiving
more and more responsibility for control (acting).
The challenge with wireless networks is that they use a communication medium which is
intrinsically open and vulnerable to interference, so that the medium access
control (MAC) plays a vital role for the guarantee of
end-to-end qualities \cite{Survey:MissionCritical}. The two industrial wireless
standards that stand for the largest market share, WirelessHART in the first
place, and ISA 100.11a in the second place, follow a centralized approach with a
central network manager, responsible for the scheduling configuration.
WirelessHART adopts the IEEE 802.15.4 PHY layer and the MAC needs to conform
to the IEEE 802.15.4 standard. To increase the reliability, WirelessHART uses time-division multiple access (TDMA)
and automatic repeat-reQuest (ARQ), but leaves much of the scheduling algorithm specifics and implementation to the vendor.
Contention-free protocols are preferred over
contention-based protocols, such as carrier sense multiple access (CSMA),
because packet delivery is more predictable and energy can easier be saved ({\em e.g.} by duty-cycling) \cite{bianchi2000performance, lennvall2008comparison}.

Much of the MAC-layer scheduling research that has been conducted assumes in the
network model that packet delivery is certain \cite{Ergen:2010},
\cite{Incel:2011}. The main challenge under that assumption is to identify as
short a conflict-free schedule frame as possible that allows all sensors to
transmit their packets to a set of sinks or gateways with direct and reliable connections to
the network manager.
With harsh network-wide delay and reliability constraints, the difficulty of the
scheduling problems to be solved increases drastically, and the review of the
state-of-the-art declares the combined end-to-end guarantee of quality of
service (QoS) features such as reliability and delay as open research
\cite{Survey:MissionCritical}.

In communication systems, end-to-end reliability is commonly defined as the
end-to-end probability that all transmitted packets arrive at their destination
before a common deadline has been reached, or before several individual
deadlines have been reached.
For many industrial applications reliability guarantees are a fixed constraint.
A latency bound is of no value if it does not hold at a certain reliability
level and lost packets or missed deadlines imply severe consequences for
many industrial control applications.
Scheduling algorithms in the literature lack the ability of
reliability guarantees in the schedule creation process \cite{Ergen:2010, yan2014hypergraph}.
An improvement of reliability can be achieved by means of redundancy.
Redundancy can be implemented on different layers, and come in the form of, {\em e.g.}, redundant packet content and error correction codes
(physical layer), repeated transmissions (MAC layer), the installation of
relays to improve connectivity, or the concurrent transmission of packets over
multiple paths. Every reliability improving means
has it's caveat, including increased costs, maintenance efforts, or
packet delay.

In this paper, we improve reliability
on the MAC layer by means of dedicated spare opportunities for packet
transmission.
The authors in \cite{yan2014hypergraph} propose a two-step process where
initially a schedule is created by one of many scheduling algorithms and then incremented
in order to maximize reliability until the maximum acceptable delay has been
reached. However, the proposed incremental improvement function is not scalable,
and it requires a two-step process; scheduling, and extending.
The problem we address is the lack of a scalable, fast to compute and
reliability ensuring scheduling algorithm. Running solver software for hours to
identify optimal results is not an option for industrial WSN with resource constrained field devices.

We propose SchedEx, a low-complexity generic extension for existing slot-based
scheduling algorithms. SchedEx ensures that the resulting schedule frame
guarantees a lower end-to-end reliability bound $\underline\rho$. SchedEx, as
opposed to the \emph{Incrementer} presented in \cite{yan2014hypergraph}, does not require a valid schedule frame to start with. With SchedEx, schedules can be swiftly calculated at the network
manager, or the sensors. The information can be used for a substantially faster
re-scheduling than what has been possible thus far. Further, the algorithm could be integrated into cross layer optimization frameworks for re-routing and
scheduling. We prove that SchedEx guarantees a user-defined reliability level
$\underline\rho \in [0,1]$, not increasing the scheduling algorithm complexity
by more than a constant factor.
The repetition vector it requires for the calculation of necessary repeats is created in
$O(|\mathcal{T}|)$, where $\mathcal{T}$ is the amount of transmitters in
the network. We compare four recently introduced scheduling algorithms using
SchedEx and obtain results that are competitive to the method in
\cite{yan2014hypergraph} for harsh reliability requirements, doing that an order of magnitude faster.
We apply single-path routing with a fixed set of transmitters and a single sink,
in order to make our results comparable to existing work in
\cite{yan2014hypergraph}. However, without loss of generality, SchedEx can be
applied to topologies with arbitrary flows and/or multi-path routing.

The remainder of the paper is organized as follows. Section \ref{background}
gives an overview of related work, Section \ref{theory} provides the
network model and general assumptions. Section \ref{method} explains the
approach, specifying the investigated topologies, and introducing SchedEx. The
simulations with their results are presented in Section \ref{experiments}. The
paper concludes with a discussion in \ref{discussion} and concluding remarks in
\ref{conclusion}.
\section{Related Work}
\label{background}
Various scheduling problem variants have been investigated in the literature
({\em e.g.} \cite{Ergen:2010,Incel:2011,Djukic:2009, shen2013sas}). The decision making with
respect to medium access and routing have a dominating influence on the QoS.
The medium access in proposals for mission-critical applications have been
dominated by the TDMA approach \cite{Ergen:2010}, also adapted in WirelessHART \cite{WirelessHART}, ISA
100.11a \cite{isa2009100}, WIA-PA \cite{wiapa}, and IEEE 802.15.4e
\cite{IEEEstde}.
Using TDMA, sensors require time synchronization, and permissions for
transmission are assigned based on dedicated slots in a recurrent TDMA
schedule-frame.
That way, within-network interference is reduced to a minimum while introducing an administrative overhead.
Finding minimal TDMA schedules for a WSN graph is
an NP-hard problem \cite{Ergen:2010}.
For dynamic environments with conflicts changing over time, the tracking of an
optimum becomes computationally infeasible even ignoring interference,
packet-loss, and time synchronization in the problem formulation.
Fortunately, WSNs for automation do usually not require optimal solutions,
but solutions that satisfy QoS constraints with highest possible reliability.

Substantial research for reliability improvement has been conducted in the field
of industrial wireless sensor networks (IWSN) \cite{tii1,tii2,tii3,tii4,
tii5}, which underlines the need for reliability ensuring scheduling.
The authors in \cite{Survey:MissionCritical} give a recent review of existing research
on QoS guarantees for mission-critical WSN applications.
The diversity of applications and their constraints with WSN as a potential
solution imply that the tailored proposals in the literature are diverse and problem dependent too, for instance with regards to topology, traffic patterns, or the assignment of responsibilities in the network.
Thus, general approaches for guarantees on reliability, latency, and
other quality measures (such as jitter, or throughput),
considering network dynamics, are highly demanded.

The approaches introduced in \cite{GinMAC} and \cite{Burst} propose offline
dimensioning with a-priori assessment of packet delivery ratios on the MAC-layer
using TDMA scheduling. Burst-error metrics are used in
order to improve end-to-end reliability where latency-bounds are reported
by repeating slots until a deadline. The greedy algorithm from \cite{Burst}
employs spatial-reuse to provide latency guarantees, a technique that allows
multiple sensors to transmit simultaneously if they are not interfering with one
another which can considerably reduce schedule frame sizes, and thereby reduce
latency.
The underlying packet reception rates for the WSNs in \cite{GinMAC, Burst} go
unreported. They could be high, showing only small traces of interference, which
would not be a generalizable assumption for industrial networks \cite{FilipsPaper}.
Distributed best-effort protocols like the one in \cite{MMSPEED} can surely improve on
performance in best-effort multi-hop networks,
but for the use in industrial networks with harsh end-to-end constraints,
best-effort is not enough and centralized solutions are needed.

Recently, researchers have started to consider the trade-off
between latency and reliability given it a theoretical basis.
A first study on scheduling for WSN with end-to-end
transmission delay guarantees and end-to-end reliability maximization can be
found in \cite{yan2014hypergraph}. The authors propose two scheduling
algorithms, Dedicated and Shared Scheduling, whereas the second algorithm is a
variant of the first. One of the main differences to existing algorithms,
such as the ones that are used for comparison in \cite{yan2014hypergraph} from
\cite{Ergen:2010}, is that they consider the node-to-node packet reception
rates by ordering the sensors according to link quality
before scheduling. The authors give a sound mathematical foundation for the
problem, both covering single-path and multi-path routing, and further
introduce an optimal schedule increment strategy, which under the assumptions
of consistent node-to-node packet delivery rates optimally extends the
schedule by repeating the most suitable slot up until a certain user-defined deadline
has been reached. The schedules are therefore locally optimal only allowing
for schedule extension by the repetition of existing slots (no combination or
change of slots allowed). In \cite{yan2014hypergraph}, Dedicated and Shared
Scheduling showed to outperform the competitors in all experiments. Single-path
routing showed to outperform any-path routing with respect to end-to-end reliability,
which shows that the assumption of using redundancy in terms of multiple routes
does not necessarily pay off, and re-scheduling can be a more lightweight, and well performing, alternative.

\section{Network Model}
\label{theory}
A WSN is defined as a digraph $G=\{V, E\}$. $V=\mathcal{T} \cup \mathcal{S}$ is the set of sensors in the network, with transceivers
$\mathcal{T}$, sinks $\mathcal{S}$, and $\mathcal{T} \cap \mathcal{S} =
\emptyset$. We model the successful
arrival of a packet as node-to-node packet arrival rates that include the
successful reception of an acknowledgement (ACK) on the receiver side before the
end of the time-slot has been reached.
Information of packet reception rates can be obtained based on empirical data by methods such as
\cite{fonseca2007four}.

$E=\{ (t,p) | q_{tp} > 0 \}$ is the set of directed links from any sensor
$t \in \mathcal{T}$ to any receiver $p \in V$, with node-to-node packet
reception rates, expressed as $q_{tp}$.
$Q$ is the link quality matrix of size $|\mathcal{T}| \times |V|$ that assigns each transceiver in row $t$ the empirical likelihood of receiving a packet in $p$ as $q_{tp}$.
We assume directed links, because link strength can vary depending on the transmission direction and
antenna positioning.

The state of the network is represented by a buffer $b = (b_1, ..,
b_{|\mathcal{T}|})$ where $b_t$ with $1 \leq t \leq |\mathcal{T}|$ is the
amount of packets waiting for transmission at sensor $t$. $b = \overrightarrow{x}$ signifies that
each sensor has equally $x$ packets in the local buffer.
 
\subsection{Model Assumptions}
All sensors operate in half-duplex mode, since full-duplex is not readily
available and still work in progress. Time-synchronization is a crucial
requirement for the enabling of predictability of QoS on a network level. We assume time synchronization
between the sensors to be guaranteed by the lower layers of the protocol stack. Imperfect time-synchronization requires receivers to listen some time before the
slot starts, and stop listening some time thereafter (\emph{guard time}).
References for the efficient realization of time-synchronization can be found in
\cite{stanislowski2014adaptive}.
Further, in a real TDMA-based setup, each sensor receives a
transmission and receiving list, outlining in which recurring time-slot of the
schedule-frame it has the right to transmit or is supposed to listen for
incoming packets to be relayed.
This \emph{duty-cycling} is a common approach for the reduction of battery
drain.
Fixed node-to-node reliabilities and topologies with a fixed amount of sensors
and sinks are assumed in this paper. Sampled data is transmitted via multi-hop
from the sensors to the sinks, and all packets have the same reliability
requirements. Packets are further always produced at the beginning of a schedule
frame $F$ and the common deadline to be held is the final slot in the schedule,
after which all sensors have to have succeeded to transmit all their packets
for the frame. Packets that could not be transmitted in $F$ are removed
from the buffer in order to preclude contention-problems.
Assuming stochastic independence between the different
events of transmission success and failure, two independent transmissions
from $t$ to $p$ and $t'$ to $p'$ succeed with a probability
{\small
\begin{eqnarray}
P_{success}(tp, t'p')=q_{tp} q_{t'p'}
\end{eqnarray}}
\subsection{Constrained Optimization Problem}
We formulate the problem as a constrained optimization problem.
Let $o : \mathcal{I} \rightarrow \mathbb{R}$
be an objective function that assigns each configuration of the network
$x \in \mathcal{I}$ a quality index.
We define an optimization problem as
{\small
\begin{eqnarray*}
\min_{x} & o(x) & x\in\mathcal{I}\\
\text{subject to } & c_i(x) & \leq 0
\end{eqnarray*}}
where all constraints $c_i$ must hold true for $x$ in order to be
considered a feasible solution $x\in \mathcal{I}$\footnote{Any equality
constraint can be reformulated into inequality constraints, and any
maximization problem can be expressed as the inverse of $o$.}. The QoS
requirements can either be expressed as constraints or be included in the
objective function. Here, the end-to-end delay given a routing table $R$ is to
be minimized:
{\small
\begin{eqnarray}
o(x) = |F|, \forall x =(R, F) \in \mathcal{I}  \label{mat:o}
\end{eqnarray}}
with $| \cdot |$ signifying the size of the frame, and $SA(R)=F$, for a given
scheduling algorithm $SA$. This formulation is similar to \cite{Ergen:2010},
with the difference that any solution $x=(R,F)$ is not only a schedule $F$ based on a heuristically rendered
routing table $R$, but identified and assessed as a whole. We further extend the
formulation by the additional end-to-end reliability constraint:
{\small
\begin{eqnarray}
c_0 : \underline{\rho} \leq \rho \label{mat:reliability},
\end{eqnarray}}
where $\underline{\rho}$ is the required and $\rho$ the actual
end-to-end reliability of the network. At times, reliability can be
improved at a low cost, for instance by the installation of relays. Other times,
latency trade-offs are easier to make, for instance where the pace of a production-line can adapt to the arrival of
  packets. Thus, the question if a problem should be viewed as a
  slot minimization or reliability maximization problem
  depends on the context of the application. In this paper, we made a case for
  scenarios in which topology changes are not an option, but where shorter
  scheduling cycles can lead to an increased productivity.
This approach is different to the one taken in \cite{yan2014hypergraph}, where reliability is part of the objective function, and the deadlines are expressed as a constraint.
In a real setting, the problem is commonly a constraint satisfaction
problem, where all quality requirements have to be guaranteed.
However, formulating the problem as an optimization problem allows for a better
comparison among different approaches and algorithms, which is one of the
objectives in this paper.

The formal constraints that ensure the validity of the solutions are introduced
below. On the network layer we assume a binary routing table $R$ of size
$|\mathcal{T}| \times |V|$ that with $r_{tp}=1$ assigns each transceiver in row
$t$ one recipient (\emph{parent}) $p$. $R$ fulfills the following constraints:
{\small
\begin{align}
c_1: &&r_{tt} = 0\label{const1:R}\\
c_2: &&\sum\limits_{p\in V} r_{tp} \leq 1\label{const2:R}\\
c_3: &&r_{tp} = 1 &\Rightarrow (t,p) \in E \label{const3:R}
\end{align}}
Thus, no sensor is its own parent
(\ref{const1:R}), any transmitter has at most one parent
(\ref{const2:R})\footnote{The model thus considers isolated transceivers.}, and
a directed route requires a link (\ref{const3:R}).
R is expected to be \emph{successful}, meaning that all transceivers in
$\mathcal{T}$ can forward their packets to a sink in $\mathcal{S}$ via the established routes using multiple hops (\emph{multi-hop routing}).

The MAC layer manages access to the shared communication medium. We
apply contention-free scheduling due to it's suitability for QoS satisfaction
(see. {\em e.g.} \cite{Survey:MissionCritical}).
A schedule or \emph{schedule-frame} $F$ is a binary matrix of size $m \times
|\mathcal{T}|$ that with $f_{st} = 1$ assigns transmitter $t\in\mathcal{T}$ transmission allowance
at discretized time $k$ in \emph{time-slot} $s = k\mod{m}$. $F$
is continuously repeated in order to ensure that the scheduled transceivers can
transmit periodically. We assume a single communication channel to be used as
commonly assumed in the scheduling literature.
A schedule $F$ is called \emph{valid} or \emph{collision-free} if:
{\small\begin{align}
c_4: && f_{st} = 1 \Rightarrow & f_{sp} q_{pt} = 0&&\label{const2:S}\\
c_5: && f_{st}f_{st'} = 1 \Rightarrow & q_{tp}q_{t'p} = 0 \text{  } \vee
\label{const4:S}\\
&&& \sum_{t \in \mathcal{T}} f_{st}r_{tp} = 0  \text{  }
\end{align}}
A sensor cannot receive at the same time as it transmits (\ref{const2:S}) and a
receiving sensor may not be disturbed by a second concurrent transmission
(\ref{const4:S})\footnote{Either no sensor exists that hears two concurrent
transmissions, or that sensor is not interested in either transmission.}. Thus,
spatial-reuse within a slot is allowed if no conflicts arise and transceivers can
be scheduled in multiple slots. The interfering range is modeled by directed packet
reception rates of $0.00001$ so that the constraints assure the avoidance
of transmissions for interfering sensors.

A pair of schedule frame $F$ and routing table $R$, $x_b=(F,R)_b$, is here
called \emph{successful} with respect to an initial buffer vector $b$ if it fulfills all the former constraints and enables
the transmission of all packets in the buffers to
any sink via multi-hop after the consecutive execution of all $m$ slots. In
other words, there is \emph{a chance} that all packets are delivered at the end of the schedule. In existing
work, the scheduling problem usually assumes an initial buffer vector
$b = \overrightarrow{1}$, a goal buffer vector
$b = \overrightarrow{0}$, and a perfect channel not effected by
packet loss \cite{Ergen:2010}.
This, in fact, reduces the problem to the minimization of schedule size, which
is NP-hard in its own respect \cite{Ergen:2010}.

\section{Guaranteeing Reliability}
\label{method}
Instead of calculating the exact reliability, as done for instance in
\cite{yan2014hypergraph}, we prove here that SchedEx guarantees a reliability lower bound $\underline\rho$. Given the quality variations of links over time,
a calculation of exact reliabilities is only of theoretical interest, while
ensuring certain lower bounds is highly relevant. The Incrementer and SchedEx
are two entirely different approaches to ensure end-to-end reliability
$\underline\rho$.
The Incrementer optimally improves a schedule frame $F$ by repeating the slot
within the frame which incurs the largest reliability gain until
$\underline\rho$ is ensured. It therefore requires a valid schedule frame $F$.
SchedEx is a low-complexity generic extension for existing
scheduling algorithms which ensures that the resulting schedule frame guarantees
a lower end-to-end reliability bound $\underline\rho$. SchedEx does not require a
schedule frame $F$ to start with. We show that
the produced schedules are competitive to the ones produced by the Incrementer introduced in \cite{yan2014hypergraph}, and are calculated more than an \emph{order of magnitude} faster. In the following, we are gradually increasing a network model from one transmitted packet over a single link to $k_1,\ldots,k_l$ transmitted packets over $l$ independent links, eventually forming a whole network, which guarantees the end-to-end reliability
$\underline\rho$.
\subsection{1 Link, 1 Packet}
Assuming a Bernoulli Distribution
for the arrival success of a packet submitted from a sender $s$ to a receiver
$r$, the node-to-node probability that a packet arrives in a series of $n$
attempts can be calculated by:
{\small
\begin{align}
P(success \geq 1 | n \text{ attempts}) &=\sum_{k=1}^n P(success = k | n \text{
attempts})\\
&=1-P(success = 0 | n \text{ attempts})\\
&=1 - (1-q)^n \label{successProb}
\end{align}}
The success probability is empirically explained by the
packet reception rate $q$ which translates into the success
probability for the individual, stochastically independent, attempts.
For mission-critical applications, a minimum reliability $\underline\rho =
P(success \geq 1 | n \text{ attempts})$ must be ensured. Using mathematical
terms, $\underline\rho$ is called a lower bound. For a given $\underline\rho$,
we require $n$ attempts to ensure (\ref{mat:reliability}).
\begin{prop}
\label{prop:1link1packet}
A lower bound of required attempts $n$ that ensures at
least one transmission to succeed over a link with packet reception rate $q$ and
at reliability level $\underline\rho$ can be obtained by
{\small
\begin{eqnarray}
n = \left\lceil \frac{\log(1-\underline\rho)}{\log(1-q)} \right\rceil \label{repetitions}
\end{eqnarray}}
\end{prop}
\begin{proof}
We activate (\ref{successProb}) for $n$, which leaves us with:
{\small
\begin{eqnarray}
n' = \frac{\log(1-\underline\rho)}{\log(1-q)}
\end{eqnarray}}
Repetitions must
be integer, thus, the amount of repetitions must be the next larger integer in order to ensure that $n$ is not
underestimated. It follows the proposition that $n$ is a lower bound that
ensures $\underline\rho$.
\end{proof}
It can further be proven, by contradiction, that $n$ is in fact the
lowest lower bound by differentiating the two cases where $n'$ is
or is not integer.
\subsection{1 Link, $k$ Packets}
The first generalization of the network addresses the transmission
of $k$ packets over a single link.
\begin{prop}
\label{prop:kPackets}
For $k$ packets to be independently transmitted from sender $s$ to receiver $r$,
a lower bound of required attempts $n_i$ per packet that ensures a 
reliability $\underline\rho$ for all of the $k$ packets to
arrive can be obtained by
{\small
\begin{eqnarray}
n_i = \left\lceil \frac{\log(1-\underline\rho_i)}{\log(1-q)}
\right\rceil
\label{repetitions2}
\end{eqnarray}}
with
{\small\begin{eqnarray}
\underline\rho_i = \underline\rho^{\frac{1}{k}}
\end{eqnarray}}
\end{prop}
\begin{proof}
In order to ensure reliability $\underline\rho \in [0,1]$, we know that $n$
repetitions are required for 1 packet to arrive, according to
(\ref{repetitions}) to ensure (\ref{mat:reliability}). For $k$ stochastically
independently transmitted packets, we must find
$\underline\rho_1,\ldots,\underline\rho_k$ that ensure both (\ref{repetitions2}) and
{\small\begin{eqnarray}
\underline\rho \leq \prod_{i =1}^k \underline\rho_i \label{mat:relProd}
\end{eqnarray}}
For $\underline\rho_i \in [0,1], i \leq k$, (\ref{repetitions2}) holds
trivially due to (\ref{repetitions}). Without loss of generality, we choose
$\underline\rho_i = \underline\rho_j = \underline\rho^\frac{1}{k}, \forall i,j \
\leq k$. Since $\underline\rho$ must be in $[0,1]$, $\underline\rho_i$ is too.
We show (\ref{mat:relProd}) by
{\small\begin{eqnarray}
\prod_{i=1}^k \underline\rho_i = \prod_{i=1}^k \underline\rho^\frac{1}{k} =
(\underline\rho^\frac{1}{k})^k = \underline\rho
\end{eqnarray}}
\end{proof}

Since $q$ and $\underline\rho_i$ are identical for all $k$ transmissions, is
$n_i$, the number of attempts per transmission, identical for all $k$
transmissions too.
We therefore claim the total amount of attempts to ensure $\underline\rho$ to be
{\small\begin{eqnarray} n = k \cdot n_i
\end{eqnarray}}
\subsection{$l$ Links, $k_1,\ldots,k_l$ Packets}
We further generalize the guarantee to $l$ senders $t_1,\ldots,t_l$ with link
qualities $q_1,\ldots,q_l$ and $k_1,\ldots,k_l$ packets, and
reliability demand $\underline\rho$.
\begin{prop}
\label{prop:oneReceiver}
Given the transmitters $t_1,\ldots,t_l$, the packet
reception rates $q_1,\ldots,q_l$ to a receiver $r$, and  $k_1,\ldots,k_l$
packets for a stochastically independent transmission at each transmitter, a
reliability of $\underline\rho$ for all packets to arrive in $r$ can be guaranteed by $n_i$ repetitions for each
packet, according to:
{\small\begin{eqnarray}
n_i = \left\lceil \frac{\log(1-\underline\rho_i)}{\log(1-q_i)} \right\rceil
\label{repetitions3} \label{math:n2n}
\end{eqnarray}}
with
{\small\begin{eqnarray}
\underline\rho_i = \underline\rho^{\frac{1}{l\cdot k_i}}
\end{eqnarray}}
\end{prop}
\begin{proof}
Because of the demand $\underline\rho \leq \rho$ and the independence assumption, the constraint
\begin{eqnarray}
\underline\rho \leq \rho = \prod_{i=0}^l \rho_i \label{mat:proof3}
\end{eqnarray} 
has to be fulfilled where $\rho_i$ is the reliability of link $i$.
Similar to the proof for (\ref{repetitions2}), each link ensures
reliability demand $\underline\rho_i \leq \rho_i$ according to
(\ref{mat:proof3}), with $\rho_i \geq \underline\rho^{\frac{1}{l}} =
\underline\rho_i$.
Applying Proposition (\ref{prop:kPackets}) for link $i$, we get that
$\underline\rho_i^\frac{1}{k_i}$ guarantees the reliability
$\underline\rho_i $ for each packet by applying $n_i$ repetitions.
It follows that
\begin{eqnarray}
\underline\rho_i^\frac{1}{k_i} =
(\underline\rho^\frac{1}{l})^\frac{1}{k_i} = \rho^\frac{1}{l\cdot k_i}&, k_i>0
\end{eqnarray}
which concludes the proof.
\end{proof}
The total number of required attempts is then calculated by
{\small\begin{eqnarray} 
n = \sum_{i=1}^l k_i n_i
\end{eqnarray}}
with $n_i$ from Proposition (\ref{prop:1link1packet}). Further, because of the
stochastic independence assumption for all transmissions, is the reliability
$\underline\rho$ that all packets arrive at different receivers $r_1,\ldots,r_l$ for transmitters  $t_1,\ldots,t_l$  equal to the former
proposition, with $q_1, \ldots, q_l$ defined for the respective sender, receiver
pair (see Figure \ref{fig:kSenders}).
\begin{center}
\begin{overpic}[width=\columnwidth]{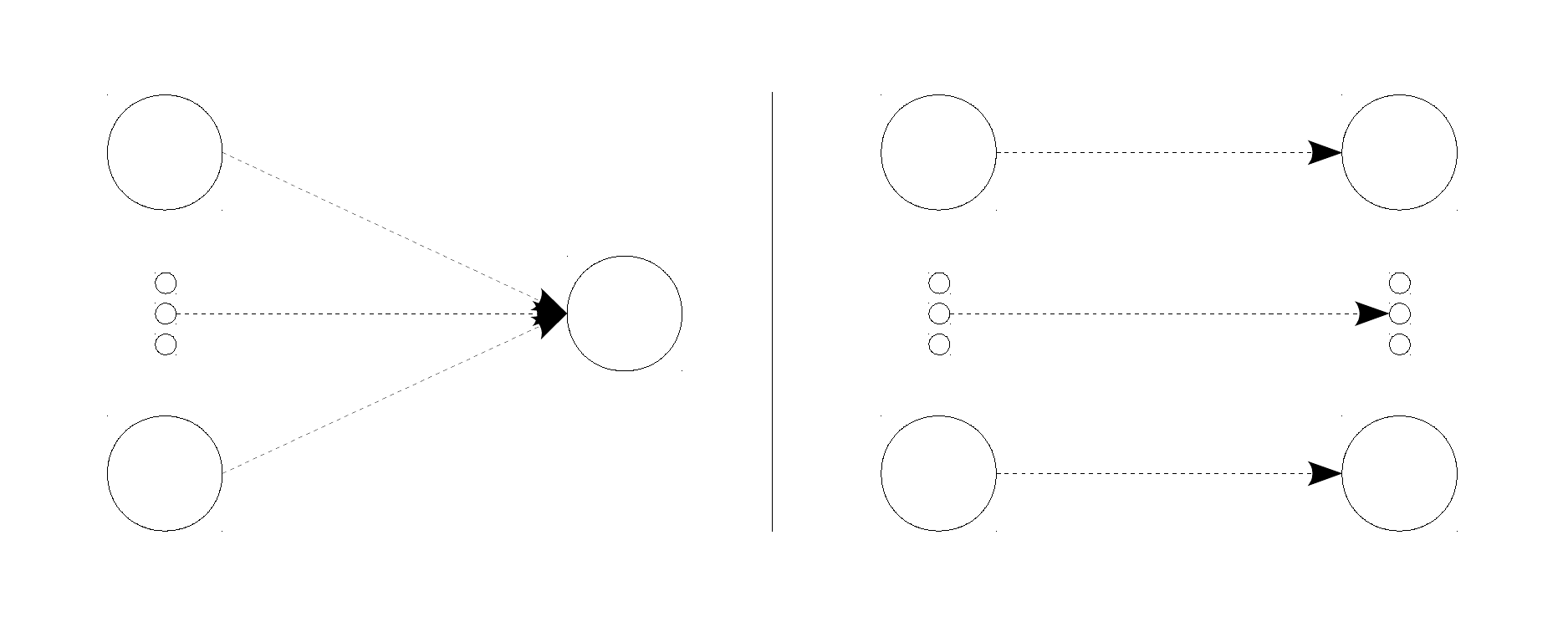}
 \put (9.5,29.5) {\large$\displaystyle{t_{1}}$}
 \put (5.5,25) {$\small\displaystyle{k_{1}}$}
 \put (9.5,9) {\large$\displaystyle{t_{l}}$}
 \put (5.5,4.5) {$\small\displaystyle{k_{l}}$}
 \put (39,19.5) {\large$\displaystyle{r}$}
 \put (24, 27) {\large$\displaystyle{q_{1}}$}
 \put (20, 21.5) {\large$\displaystyle{q_{i}}$}
 \put (24, 12) {\large$\displaystyle{q_{l}}$}
 
 \put (58.5,29.5) {\large$\displaystyle{t_{1}}$}
 \put (55,25) {$\small\displaystyle{k_{1}}$}
 \put (58.5,9) {\large$\displaystyle{t_{l}}$}
 \put (55,4.5) {$\small\displaystyle{k_{l}}$}
 \put (88,29.5) {\large$\displaystyle{r_1}$}
 \put (88.5,9) {\large$\displaystyle{r_l}$}
 \put (73, 28) {\large$\displaystyle{q_{1}}$}
 \put (73, 17.5) {\large$\displaystyle{q_{i}}$}
 \put (73, 7.5) {\large$\displaystyle{q_{l}}$}
\end{overpic}
  \captionof{figure}{Left: Transceivers $t_1,\ldots,t_l$ with the same
  destination $r$ for their $k_i$ packets (Proposition
  3). Right: The $l$ transceivers have different
  destinations $r_1,\ldots, r_l$.
  The reliability $\underline\rho$ is in both cases guaranteed at $r$
  ($r_1,\ldots,r_l$) with $n_i$ repetitions according to (\ref{math:n2n}).}
  \label{fig:kSenders} 
\end{center}
\subsection{Entire Network}
In order to generalize the model for end-to-end reliability $\underline\rho$ of
an entire network, we first introduce some restrictions. First, given a schedule
frame of length $m$, we assume that all packets to be transmitted are available
at their source at schedule slot 1, and that no new packets are created during
the schedule execution, or that those packets are enqueued for the next
schedule frame. $\underline\rho$ is the guarantee that all packets arrive
until the end of the schedule frame, after the execution of all $m$ slots.
\begin{prop}
\label{prop4}
Given a network with a single-path routing table $R$, sinks $\mathcal{S}$, 
transceivers $\mathcal{T}$, a link quality matrix $Q$,
and a total amount of $k_1,\ldots,k_{|\mathcal{T}|}$ packets passing each
transceiver $t, 1 \leq t \leq \mathcal{T}$. on their way towards any of the
sinks via multi-hop. An end-to-end reliability of $\underline\rho$ can be
guaranteed if $n_t$ repetitions are executed for each route from $t$ to $ R_t$ for each packet over the whole network, according to:
{\small\begin{eqnarray}
n_t = \left\lceil
\frac{\log(1-\underline\rho^\frac{1}{|\mathcal{T}| \cdot
k_t})}{\log(1-q_{t R_t})}\right\rceil , \forall t \in \mathcal{T}
\label{repetitions4}
\end{eqnarray}}
and
{\small\begin{eqnarray}
n = \sum_{t\in \mathcal{T}} k_t n_t. \label{mat:repCalc}
\end{eqnarray}}
\end{prop}
Here, $R_t$ is the parent of $t$ in routing table $R$.
The proof is identical to Proposition (\ref{prop:oneReceiver}), with the
difference that the amount of links for a single-path routing table is $|\mathcal{T}|$ and the amount of
packets $k_t$ to be transmitted over link $r_{t{R_t}}$ is the sum of
all packets produced in the sub-tree of $R$ with root $t$.
\subsection{SchedEx}
Scheduling algorithms in the literature ({\em e.g.}
\cite{Ergen:2010, yan2014hypergraph}) follow the same general pattern:
\begin{itemize}
  \item \textbf{While} not all packet-buffers empty
  \begin{enumerate}
  \item \textbf{apply scheduling algorithm} to decide the next slot(s)
  \item \textbf{append the slot(s)} to the schedule frame 
  \item \textbf{update packet buffers} according to the transitions
  \item update meta-data (if required)
\end{enumerate}
\end{itemize}
SchedEx, in Algorithm 1, implements the packet buffer update in step 3 for
each scheduled transceiver $t  \in \mathcal{T}$ to ensure end-to-end reliability
by a controlled packet move delay.
The algorithm maintains a vector of down-counters for all transceivers $t \in \mathcal{T}$. We introduce the vector $\overrightarrow\tau$, which for each transceiver $t\in \mathcal{T}$ defines the required number of attempts $\tau_t
= n_t$ to ensure the link-reliability according to Proposition
(\ref{prop4}), and initialize counter $c_t$ by $\tau_t$. Each time a transmission attempt in
$t$ has been registered in a slot, the counter for $t$ is decremented. The
scheduler moves a packet from the transmitter buffer to the receiver buffer once
the counter has been hit $\tau_t$ times. Thereafter, the counter is
reinitialized by $\tau_t$. This procedure that implements redundancy into the
schedules, guarantees a reliability not below $\underline\rho$ since
$\overrightarrow\tau$ is created from (\ref{repetitions4}) for each transition.
\begin{algorithm}
\algsetup{linenosize=\small}
\small
\begin{algorithmic}[1]
\FUNCTION $updatePacketBuffers(t, \mathcal{S}, b, c, \tau)$
\REQUIRE Scheduled Transmitter $t\in \mathcal{T}$, Sinks $\mathcal{S}$,
Buffer State $b$, Counter Vector $c$, Repetition Vector $\tau$
\IF{$b_t > 0$}
\STATE $c_t\text{-}\text{-}$ \COMMENT count down attempts for transmitter $t$
\IF{$c_t = 0$}
\STATE $b_t\text{-}\text{-}$ \COMMENT remove packet from transmitter buffer
\IF{$R_t \notin \mathcal{S}$}
\STATE $b_{R_t}\text{++}$ \COMMENT add packet to receiver buffer
\ENDIF
\STATE $c_t \leftarrow \tau_t$ \COMMENT update counter
\ENDIF
\ENDIF
\label{alg:update}
\end{algorithmic}
\end{algorithm}

SchedEx is utilized as follows. The chosen scheduling algorithm $SA$, extended
by reliability-awareness through SchedEx (Algorithm 1), executes on the central
manager (\emph{gateway}). The link quality matrix $Q$ with lower bounds is
derived ({\em e.g.} using \cite{fonseca2007four}), and routing table $R$ is created using an
arbitrary routing algorithm. SchedEx is run once at network start-time to
derive a schedule solving the optimization problem in (\ref{mat:o}). This
schedule is then deployed. SchedEx is re-run in case of routing and link
quality changes that violate the end-to-end reliability constraint $\underline\rho$.
A network can be separated into sub-networks with SchedEx running co-located, but the results presented in this
paper assume a centralized approach.

\section{Simulations}
\label{experiments}
The simulation setup is summarized in Table \ref{tab:setup}.
Since the results in \cite{yan2014hypergraph} suggest that single-path routing
results in a slightly better performance than any-path routing, and considering
the complexity increase of the problem model as well as for the deployment of
scheduled any-path routing in real world, we concentrate on single-path routing
in this paper. For comparison reasons, we take the approach from
\cite{yan2014hypergraph} where the routing table is determined by the Dijkstra shortest path algorithm applying the
expected transmission count (ETX) metric \cite{de2005high}.
Extensions of ETX using signal-to-interference-plus-noise ratio (SINR) were not considered in this paper due to the deficiencies of received signal strength (RSS), the quantifier of received signal strength in all IEEE 802.15.4-compliant devices. The hardware often cannot provide reliable readings. For an in-depth investigation in industrial environments, see \cite{barac2014ubiquitous}.
All experiments were conducted on a stationary computer with 4 Intel Xeon processors and 6 gigabyte of random access memory.
\begin{table}
\centering
\begin{tabular}{l|r}
\textbf{Parameter}&\textbf{Value}\\
\hline
Routing&Single-Path\\
Routing Algorithm&Shortest Path (ETX)\\
Network Size&$\{50, 200\}$\\
Channel Model&Rayleigh Fading\\
$a_n$&$67.7328$\\
$g_n$&$0.9819$\\
$\gamma_{pn}$&$4.2935$\\
$\alpha$&$3.3$\\
$\lambda$&$0.5$\\
SNR ($\gamma_n$)&$60$dB ($50$dB)\\
Transmission Range $r_t$&$30$ units\\
Interference Range $r_i$&$60$ units\\
\end{tabular}
\caption{Simulation setup parameters.}
\label{tab:setup}
\end{table}

\subsection{Topologies}
We reconstruct the simulation model from \cite{yan2014hypergraph} where $n$ nodes are
randomly distributed around the single sink within a circular area by a radius
of 100 units. For conformity, the first $\frac{n \lambda}{\lambda +1}$ nodes are
distributed uniformly at random within the inner circle of radius
$\frac{100}{\sqrt 2}$, with $\lambda \in [0,1]$. The remaining nodes are
randomly distributed in the outer partition of the circle. $k$ among the $n$ nodes are randomly assigned to
be the sources, each with one packet waiting to be transmitted to the sink. We model
the node-to-node packet reception rates using Rayleigh fading
\cite{liu2004cross} with the average packet loss formula reported in \cite{yan2014hypergraph} and identical
model parameters ($a_n=67.7328, g_n=0.9819, \gamma_{pn}=4.2935$).
Paper \cite{yan2014hypergraph} reports on worst-case algorithm complexity.
We further include an analysis of runtimes and performance variances to our
investigations, because execution times are often far from worst-case.
The transmission range $r_t$ is set to $30$ units, and the interference range
$r_i$ to $60$ units, while $\lambda$ is fixed to $0.5$ for all topologies.
The amount of sources with a waiting packet is set to $k=|\mathcal{T}|$.
We choose the signal to noise ratio $\gamma_0$ for the Rayleigh model to be $60$
dB, as assumed in a simulation series in \cite{yan2014hypergraph}.
All heuristic results are reported with algorithm runtime
and standard deviation over 10 repetitions. The reliability constraint $\underline\rho$ is
investigated for $\underline\rho \in \{0.9, 0.999, 0.99999\}$. The 20 randomly
created networks of sizes 50 and 200 (10 each), including their
calculated link quality levels are publicly available\footnote{\url{https://github.com/feldob/wsnScenarios}}.
For industrial networks, around 50 nodes can be anticipated
realistically, as discussed in \cite{book:iwsn}, which is why we focus on
relatively small numbers.
\subsection{Scheduling Algorithms}
We utilize the scheduling algorithms compared in \cite{yan2014hypergraph},
namely Node-based Scheduling, Level-based Scheduling, Dedicated Scheduling,
Shared Scheduling following the same setup. Node-based and Level-based
scheduling are both coloring-based algorithms. Given the conflict matrix, they
build a conflict graph of the network, assigning conflicting sensors distinct colors in order to ensure that they do not appear in the same slot for concurrent transmission, thus precluding within network interference.
Node-based Scheduling does not assume any order of the sensor assignment,
whereas Level-based Scheduling always schedules sensors close to the sink
first, and those furthers afar last. We assume that the coloring
algorithm for Node-based and Level-based Scheduling used in
\cite{yan2014hypergraph} is the basic algorithm presented in \cite{Ergen:2010}, even though this information cannot be found in
\cite{yan2014hypergraph}.\footnote{The choice of coloring algorithm does not
impact the conclusions in this paper, because all SchedEx and Incrementer
simulations are run under the same assumptions.} Dedicated Scheduling sorts the
sensors according to their reliability, and therefore ensures that stable links are scheduled first which leads to a better expected performance of the network.
The algorithm sequentially adds new sensor to the current slot, until a conflict
appears, which leads to a new slot being added to the graph with the conflicting
sensor. In its version described in \cite{yan2014hypergraph} it dedicates
each scheduled transmission to packets from a dedicated source, therefore its
name. Shared Scheduling allows the transmitter to share slots among different
sources, still reducing this to a well-defined set of (in the case of the paper)
two sources. Shared Scheduling repeats slots for valid concurrent
transmissions, because the reliability of sending $k$ packets in $n$ slots is
higher than the reliability of $k$ times one submission in one slot.
\subsection{Incrementer Approach}
We compare SchedEx to the Incrementer algorithm introduced in
\cite{yan2014hypergraph}, which we further refer to as the \emph{Incrementer}.
The Incrementer requires a valid schedule with respect to the constraints c1-c5
in Section \ref{theory}. It further requires that each sensor with transmission
rights in a slot dedicates the transmission to packets from a pre-defined set of
potential sources.
In \cite{yan2014hypergraph}, the algorithm's stop criterion is a
maximum latency bound and the objective is to maximize the reliability for that bound.
The resulting schedule is provably locally optimal, assuming that only existing
slots in position $i$ in the schedule may be repeated in position $i+1$, no
existing slot is altered, and no new slots are introduced. We implemented the
Incrementer, exchanging the stop-criterion by
end-to-end reliability demand $\underline\rho$, in order to make the algorithm
comparable to SchedEx. The Incrementer requires a valid schedule from the
respective scheduling algorithm, before it can improve on the reliability.

\section{Results and Discussion}
\label{discussion}
The bagplots in Figure \ref{fig:comp}
summarize the distributions of the simulation results for the scheduler without
the consideration of reliability, comparing the four algorithms with respect to
runtime and schedule size.
An optimal bag would be located in the lower left corner of the plot. Both,
Dedicated and Shared scheduling, substantially increase in runtime for larger
topologies and compared to the other algorithms. Node-based scheduling
throughout performs best in terms of runtime.
With respect to scheduling performance, no clear winner can be announced, all
competitors cover similar or overlapping ranges.
\begin{figure*}[htp]
  \begin{subfigure}[b]{0.5\textwidth}
                \includegraphics[width=\textwidth]{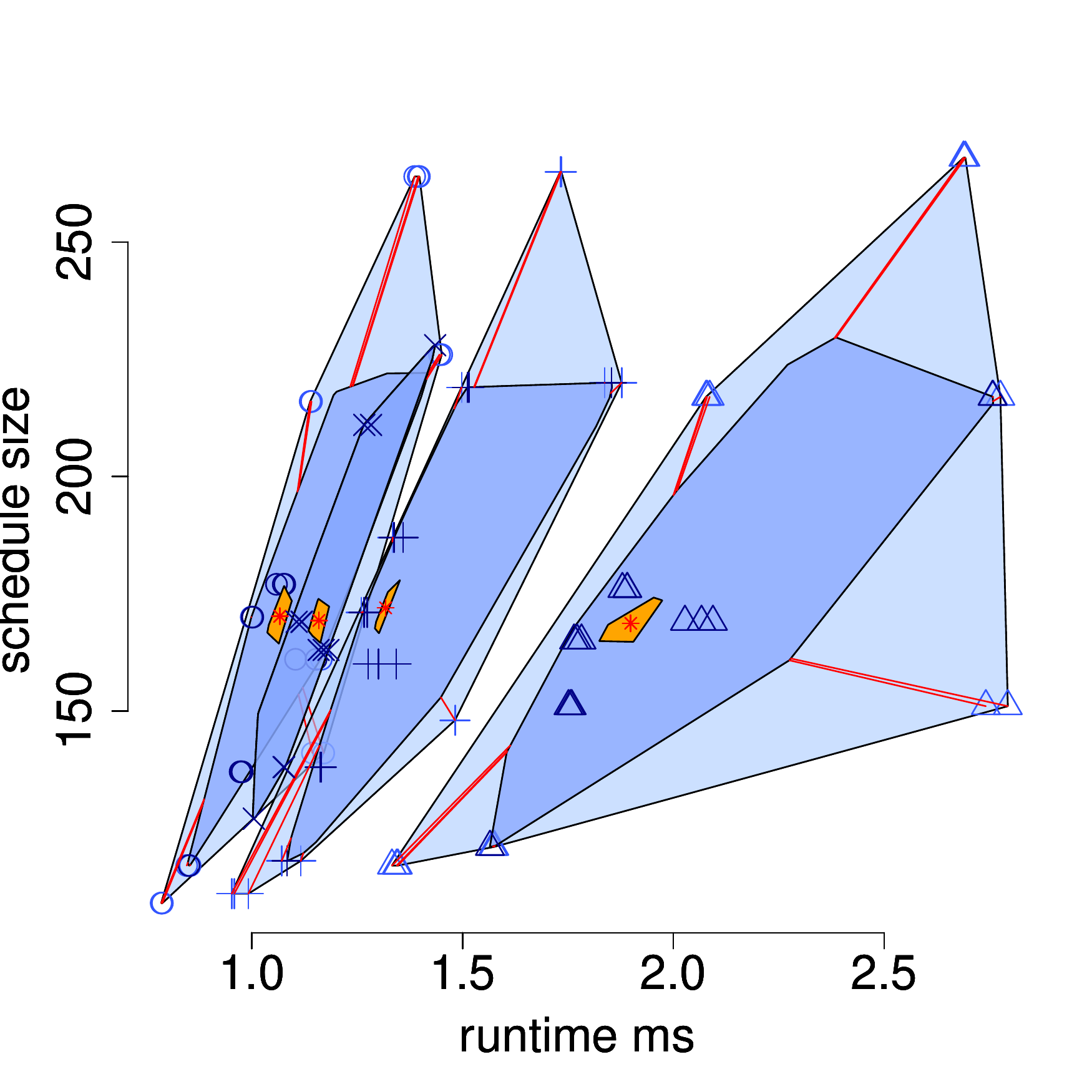}
  \caption{Topologies of size 50.}
  \label{fig:comp50}
 \end{subfigure}%
   \begin{subfigure}[b]{0.5\textwidth}
                \includegraphics[width=\textwidth]{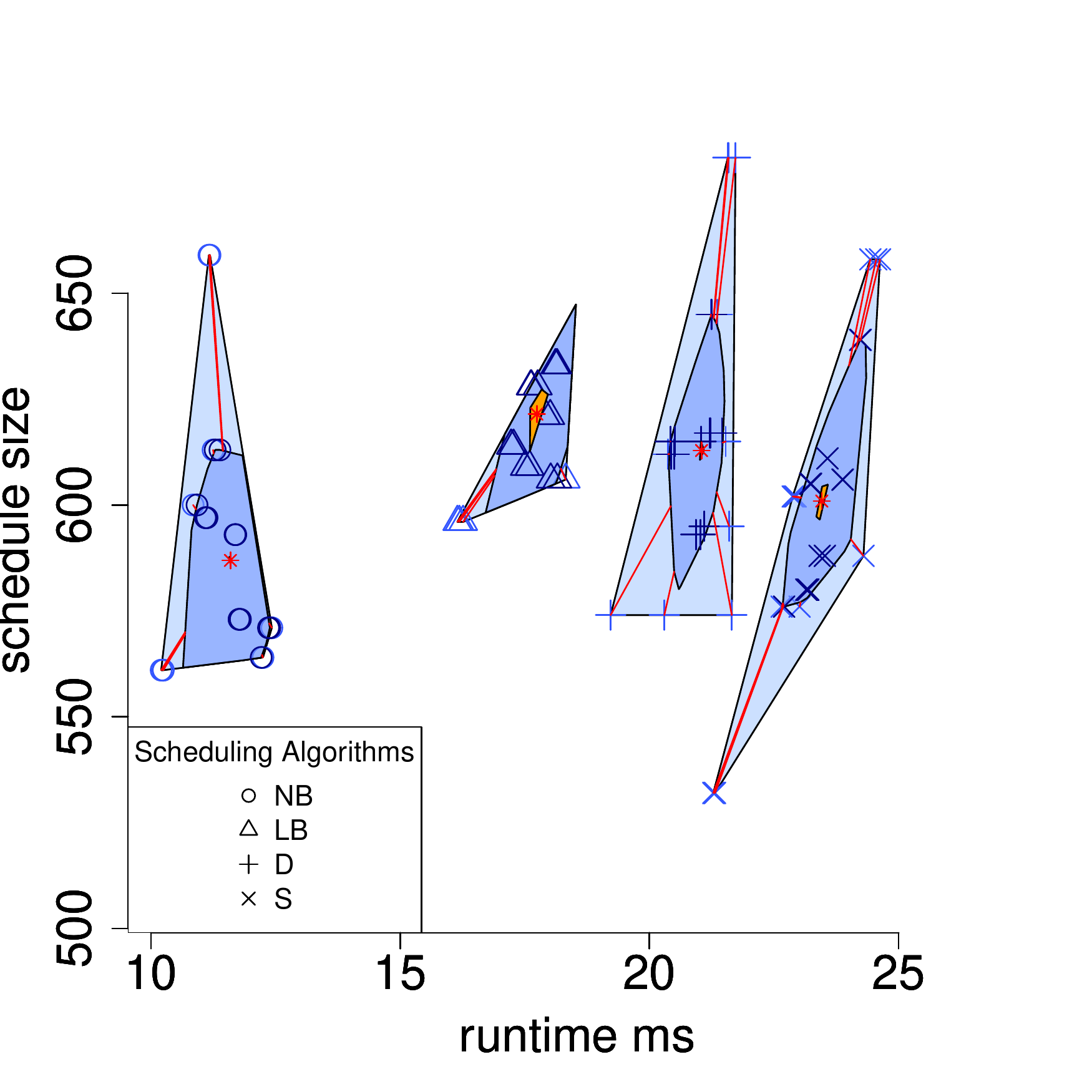}
  	\caption{Topologies of size 200.}  
	\label{fig:comp200}
 \end{subfigure}%
 \caption{Normalized bagplots comparing the scheduling algorithm performance for
 topologies of sizes 50 (left) and 200 (right). Measured runtime is plotted
 against the produced schedule size. Whereas for the smaller topologies, Dedicated and
 Shared Scheduling are only slightly slower than Node-Based
 Scheduling, do the results for the larger topologies indicate that they don't
 scale well in the size of the network. The schedule size distributions, read
 from the y-axis of the plots, overlap largely for both topology sizes, and can
 therefore not identify a winner in the comparison.}
 \label{fig:comp}
\end{figure*}
\begin{table*}
\centering
{\small
\begin{tabular}{r|l|r|r|r|r|r|r|r|r|}
\multicolumn{2}{l}{Network
Size}&\multicolumn{4}{|c|}{50}&\multicolumn{4}{|c|}{200}\\
\hline \multicolumn{2}{l}{Reliability
Algorithm}&\multicolumn{2}{|c|}{SchedEx}&\multicolumn{2}{|c|}{Incrementer}&\multicolumn{2}{|c|}{SchedEx}&\multicolumn{2}{|c|}{Incrementer}\\
\hline
$\underline\rho$&Sched. Alg.&$ms$&$|\mathcal{F}|$&$ms$&$|\mathcal{F}|$&$ms$&$|\mathcal{F}|$&$ms$&$|\mathcal{F}|$\\
\hline
0.9&Node-Based&$\mathbf{5} \pm 3$&$\mathbf{736} \pm 258$&$274 \pm 172$&$609 \pm 201$&$\mathbf{42} \pm 3$&$\mathbf{1948} \pm 102$&$2021 \pm 150$&$1640 \pm 59$\\
&Level-Based&$11 \pm 11$&$743 \pm 243$&$271 \pm 161$&$615 \pm 188$&$63 \pm 5$&$2053 \pm 93$&$2206 \pm 125$&$1742 \pm 52$\\
&Dedicated&$7 \pm 4$&$751 \pm 265$&$284 \pm 180$&$609 \pm 210$&$79 \pm 5$&$2031 \pm 101$&$2104 \pm 192$&$1668 \pm 65$\\
&Shared&$10 \pm 14$&$741 \pm 256$&$\mathbf{264} \pm 160$&$\mathbf{606} \pm 192$&$86 \pm 3$&$1962 \pm 78$&$\mathbf{1963} \pm 111$&$\mathbf{1612} \pm 41$\\
\cline{2-10}
0.999&Node-Based&$\mathbf{6} \pm 1$&$\mathbf{1083} \pm 366$&$545 \pm 318$&$987 \pm 307$&$\mathbf{64} \pm 6$&$\mathbf{2794} \pm 124$&$4090 \pm 358$&$2558 \pm 87$\\
&Level-Based&$11 \pm 3$&$1092 \pm 346$&$536 \pm 290$&$998 \pm 286$&$91 \pm 11$&$2944 \pm 121$&$4498 \pm 292$&$2712 \pm 79$\\
&Dedicated&$10 \pm 4$&$1108 \pm 375$&$544 \pm 324$&$984 \pm 318$&$111 \pm 9$&$2917 \pm 124$&$4153 \pm 404$&$2586 \pm 99$\\
&Shared&$8 \pm 2$&$1091 \pm 363$&$\mathbf{531} \pm 317$&$\mathbf{982} \pm 292$&$123 \pm 4$&$2815 \pm 82$&$\mathbf{3994} \pm 277$&$\mathbf{2511} \pm 64$\\
\cline{2-10}
0.99999&Node-Based&$\mathbf{8} \pm 2$&$\mathbf{1428} \pm 468$&$\mathbf{837} \pm 494$&$1366 \pm 415$&$\mathbf{86} \pm 6$&$\mathbf{3636} \pm 175$&$6634 \pm 667$&$3469 \pm 119$\\
&Level-Based&$13 \pm 5$&$1438 \pm 438$&$877 \pm 508$&$1380 \pm 386$&$120 \pm 7$&$3838 \pm 164$&$7459 \pm 490$&$3674 \pm 106$\\
&Dedicated&$12 \pm 5$&$1459 \pm 480$&$846 \pm 506$&$1358 \pm 428$&$144 \pm 11$&$3804 \pm 174$&$6944 \pm 756$&$3495 \pm 131$\\
&Shared&$10 \pm 2$&$1437 \pm 464$&$843 \pm 490$&$\mathbf{1356} \pm 394$&$162 \pm 7$&$3659 \pm 127$&$\mathbf{6462} \pm 543$&$\mathbf{3400} \pm 84$\\
\cline{2-10}
\hline
\end{tabular}}
\caption{Execution time in $ms$ and schedule size $|\mathcal{F}|$ for
the different scheduling algorithms under three
reliability constraints comparing SchedEx and the Incrementer from
\cite{yan2014hypergraph} with reliability stop criterion.}
\label{tab:results}
\end{table*}
See Table \ref{tab:results} for a summary of the numeric results for the
reliability ensuring simulations with SchedEx and the Incrementer.\footnote{The reported
runtimes are based on a simulator written in a high-level language, Java.
Implementing the scheduling algorithms and SchedEx in a low-level language, such
as C, may improve all runtimes by a constant factor.} Listed runtimes for the
Incrementer contain both, scheduler runtime and Incrementer runtime.
First, SchedEx is significantly faster than the Incrementer.
Figure \ref{fig:schedexTime} shows the distributions for how many times
SchedEx is faster than the Incrementer. For $\underline\rho=0.9 (0.99999)$,
SchedEx is in average $31$ ($56$) times, and thereby more than an order of
magnitude, faster.
Figure \ref{fig:schedexSize} illustrates the performance improvement
over the Incrementer in terms of schedule size. It reveals that, the harsher the
constraint, the larger the expected improvement of the schedule over the
incremental approach. For $\underline\rho=0.9$, SchedEx is expected to create
schedules that are in average $21\%$ longer than those from the Incrementer.
For $\underline\rho=0.99999$, the difference is down to $6\%$.

A signal to noise ratio of 60 dB together with a transmission range of 30 units
lead to link qualities not below 67\%, which we see as a reasonable assumption for
networks that are meant to support mission-critical applications. We would,
however, like to mention that choosing a signal to noise ratio of
50 dB leads to SchedEx being \emph{2-3 orders of magnitude} faster than the
Incrementer, while schedule sizes for the harsh scenarios with $0.99999$ lead to in average
$5.3\%$ \emph{shorter} schedules than the Incrementer. This can partly be
explained by the fact that the required schedules become much longer, and that
the Incrementer does not scale well in the size of the required schedule frame.
SchedEx performs better in the severity of the reliability constraint, and that
substantially faster.

For a given scheduling algorithm $SA$ with time complexity
$C_{Time}(SA)$, SchedEx increases the time complexity by a summand and a
factor according to:
\begin{eqnarray}
SchedEx(SA) \leq O(|\mathcal{T}|) + \tau_{max} \times C_{Time}(SA)
\end{eqnarray}
The summand $O(|\mathcal{T}|)$ explains the creation of vector
$\overrightarrow\tau$, and $\tau_{max} = \max{(\overrightarrow{\tau})}$
explains the worst case increase of a schedule.
The smaller the signal-to-noise ratio, the proportionally larger is $\tau_{max}$.

The Incrementer requires two steps. First, the initial schedule
creation of $C_{Time}(SA)$, and second, the incremental phase until the stop
criterion has been reached, according to the incremental algorithm in
\cite{yan2014hypergraph}.
The incremental phase is of complexity $O(|\mathcal{T}|D^2)$, with $D$ being
the size of the eventual frame when reaching the stop criterion, as it in each iteration
requires a recalculation of a convex improvement metric over the
iteratively increasing schedule frame $F$. Thus, the time complexity for the
Incrementer is
\begin{eqnarray}
Incrementer(SA) = C_{Time}(SA) + \Theta (|\mathcal{T}|D^2),
\end{eqnarray}
with $D$ not known a-priori. The smaller the signal-to-noise ratio,
the proportionally larger becomes $D$. Thus, where SchedEx scales linearly in
the harshness of the problem does the Incrementer scale quadratic. Further,
the average topology requires significantly fewer elementary operations
than the worst-case with SchedEx, whereas the Incrementer grows relative to
$|\mathcal{T}|D^2$.
For the Incrementer, the experiments reveal that the scheduling phase with time
complexity $C_{Time}(SA)$ stands in our simulations in average for only
about $0.003$\% of its total runtime, and does therefore not weigh in significantly.
\begin{figure}[htp]
  \begin{subfigure}[b]{0.25\textwidth}
  \includegraphics[width=\columnwidth]{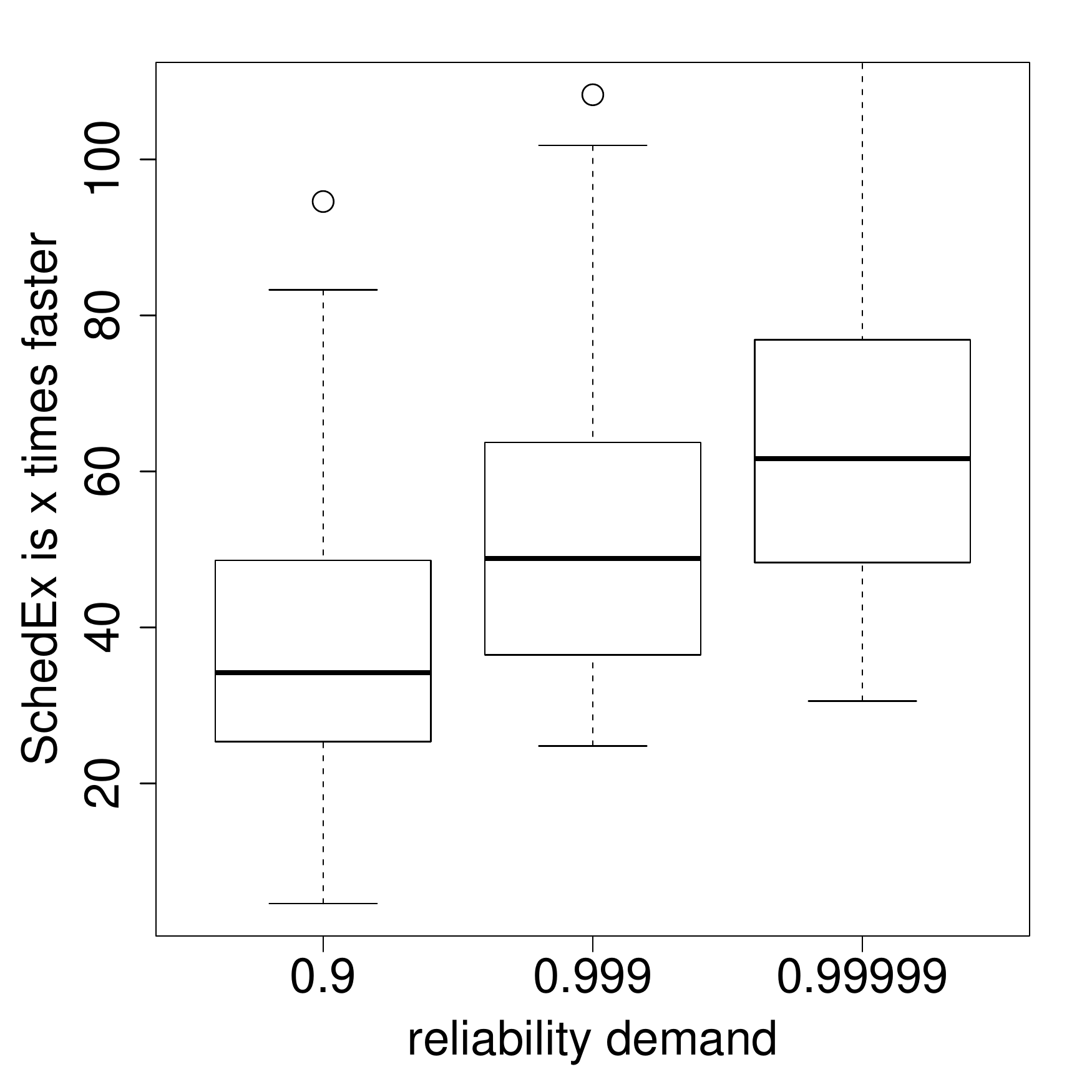}
  \caption{} 
  \label{fig:schedexTime}
 \end{subfigure}%
   \begin{subfigure}[b]{0.25\textwidth}
  \includegraphics[width=\columnwidth]{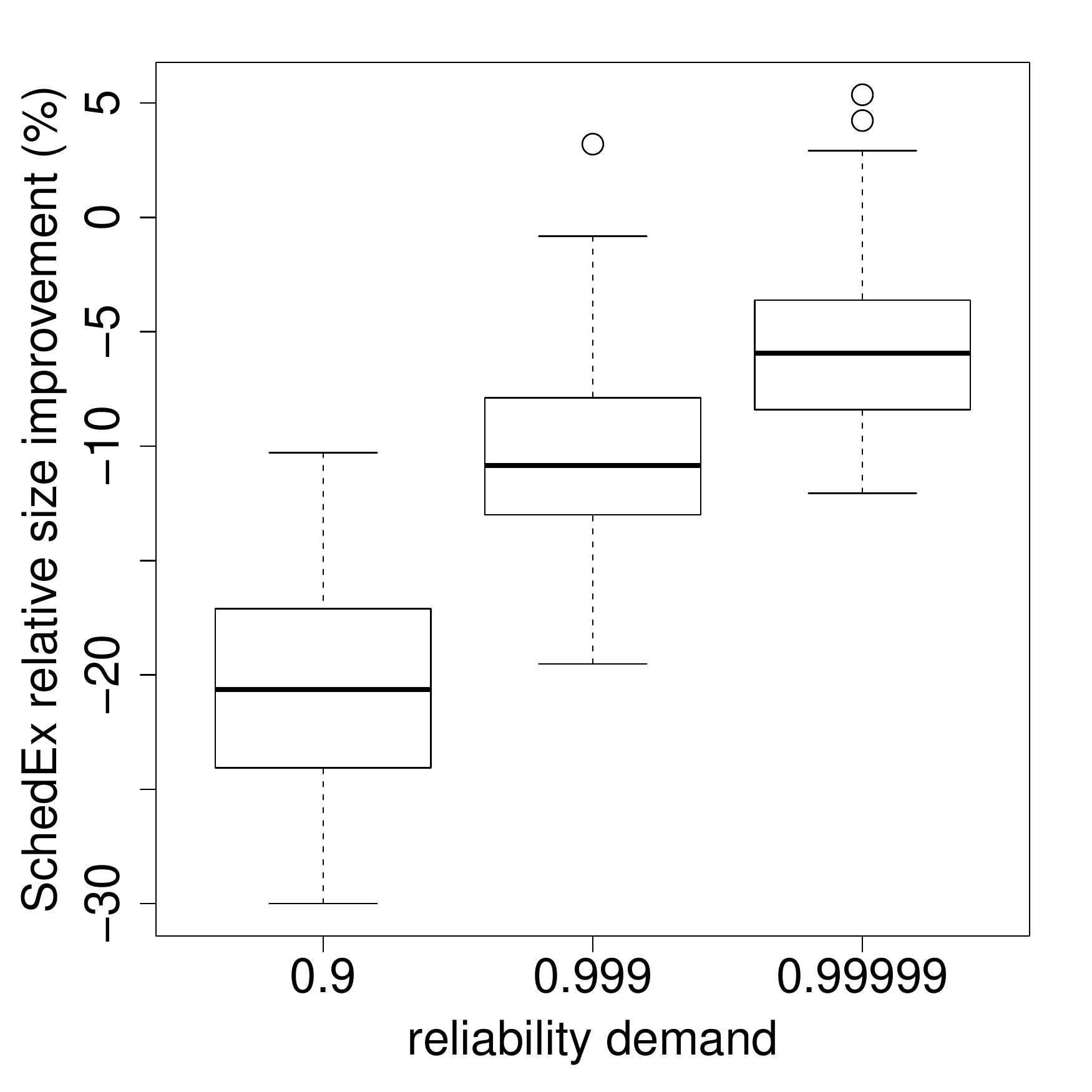}
  \caption{}
	\label{fig:schedexSize}
 \end{subfigure}%
   \caption{The percentual performance improvement of SchedEx over the
  Incrementer grouped by reliability constraint $\underline\rho$. (Left: Time,
  Right: Schedule Size)}
\end{figure}

{
\begin{figure}[htp]
\centering
\begin{overpic}[width=.8\columnwidth]{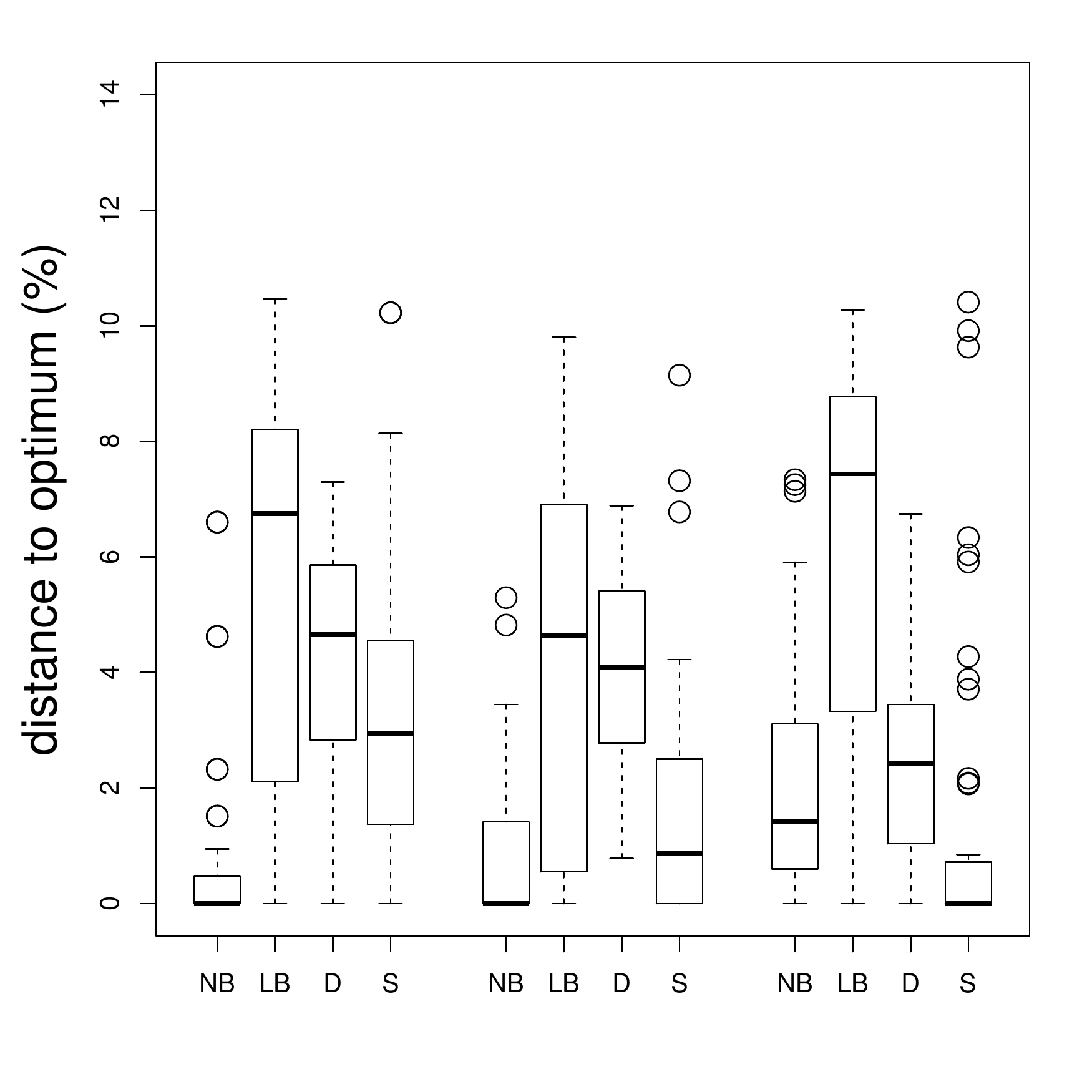}
 \put (22,3) {Basic}
 \put (45,3) {SchedEx}
 \put (70,3) {Incrementer}
\end{overpic}
  \captionof{figure}{The distance to optimum distribution for all four
  scheduling algorithms given three methods: Basic, SchedEx, and Incrementer.}
  \label{fig:relative}
\end{figure}}

\begin{figure}[htp]
\centering
  \includegraphics[width=.6\columnwidth]{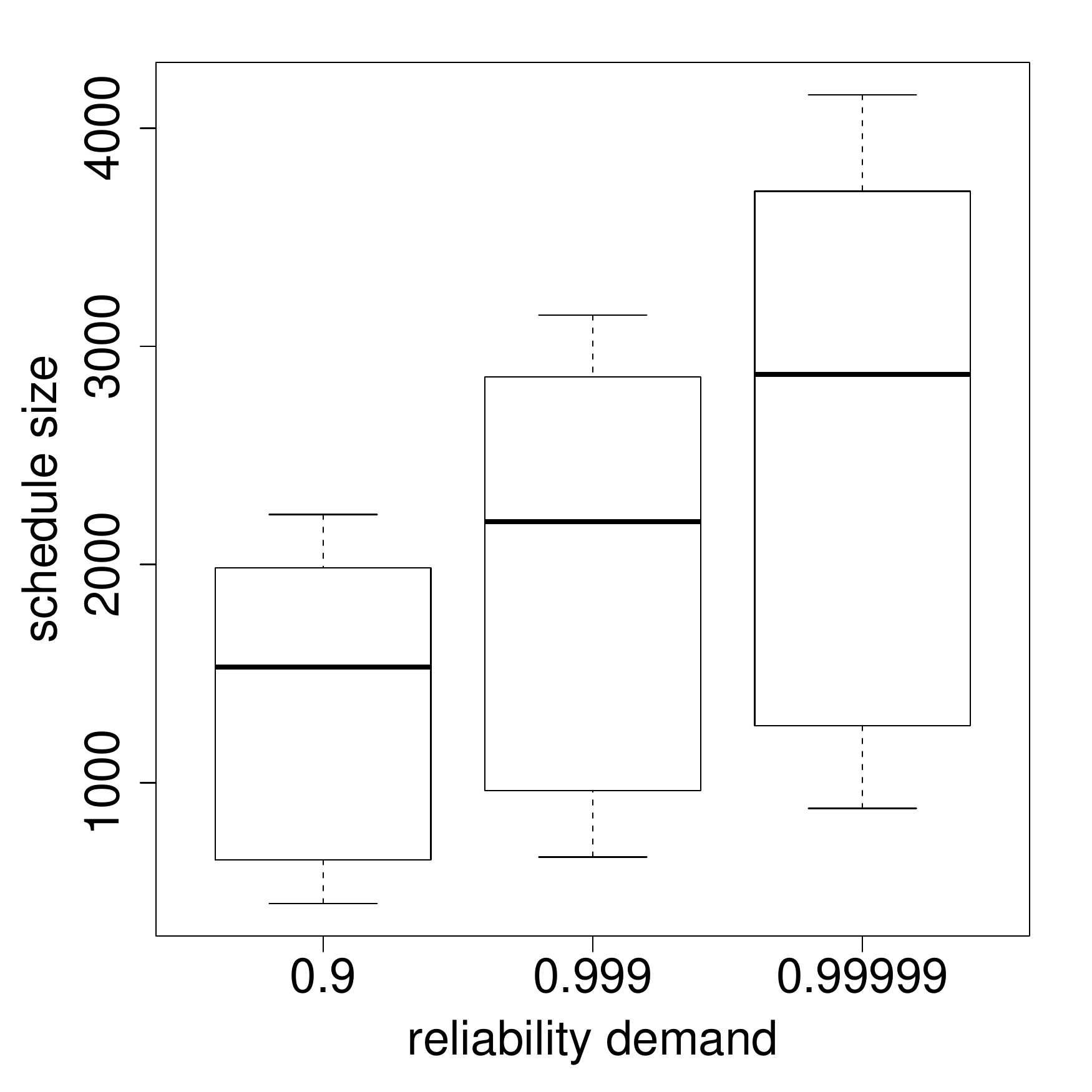}
  \caption{The difference in required schedule size that can be expected
  for varying $\underline\rho$ over all scheduling algorithms, using SchedEx.}
  \label{fig:relDemand}
\end{figure} 

Figure \ref{fig:relative} shows the relative distance to the best found solution
among the four scheduling algorithms, grouped by the basic scheduling without
reliability constraints (Basic), SchedEx, and the Incrementer. The Incrementer
favors Dedicated and Shared Scheduling, while SchedEx creates less
variance among the scheduling algorithms, suggesting that it is more generic and
less sensitive to the specifics of the scheduling algorithm of use. SchedEx
produces in average the best schedules with Node-based Scheduling, whereas the
Incrementer performs best with Dedicated Scheduling.
%

Figure \ref{fig:relDemand} illustrates the required schedule size difference for
varying $\underline\rho$ over all scheduling algorithms. The expected
schedule sizes for $\underline\rho$ being $0.9, 0.999$ and $0.99999$ are $1370$,
$1980$, and $2587$, respectively. The schedules can therefore be expected to be
ca. 45\% longer for a change in harshness from $0.9$ to $0.999$, and ca. 89\%
longer for a change from $0.9$ to $0.99999$ on the same topology.

With SchedEx, schedules with reliability guarantees can swiftly be
calculated. However, for many control applications would the latency
implications still not be acceptable. For instance, assuming WirelessHART with
slot sizes of 10 ms and $\underline\rho=0.99999$, the schedule frame would require about
1430 slots, translating into a delay of 14.3 seconds for one packet from
each transmitter to arrive at the destination on topologies of size 50. High
node-to-node packet reception rates between the sensors are one way to reduce that number,
for instance by installing relays, or using multiple channels.

\section{Conclusions}
\label{conclusion}
We introduced SchedEx, a generic scheduling algorithm extension which gives
reliability guarantees for topologies with guaranteed lower-bounded node-to-node
packet reception rates. SchedEx is an order of magnitude faster than the
theoretically well developed approach in \cite{yan2014hypergraph} and the
harsher the reliability constraint, the relatively better SchedEx performs.
Guaranteeing harsh reliability constraints implies a substantial
latency penalty, where the expected difference in schedule size from one nine
to five nines is nearly twice the length. 
As stated in the discussion, guaranteeing reliability bounds has a huge impact on the effected latency bounds.
For future work, multi-channel scheduling and multiple sink deployment should
be investigated considering the required latency improvements.
For instance, WirelessHART uses up to 15 channels
to concurrently transmit packets. Other important future scheduling extensions
are the consideration of individual latency bounds for different flows, priority
handling, and the verification of the simulation results in a testbed. The
introduced SchedEx algorithm presents one step towards scalable scheduling, a highly required feature for end-to-end QoS guarantees in WSN.

%

\bibliography{article}




\ifCLASSOPTIONcaptionsoff
  \newpage
\fi 


  
\bibliographystyle{IEEEtran}
\end{document}